\documentclass[draftcls,onecolumn,12pt]{IEEEtran}

\ifCLASSINFOpdf
   \usepackage[pdftex]{graphicx}
   \graphicspath{{../figures/}{../jpeg/}}
   \DeclareGraphicsExtensions{.pdf,.jpeg,.png}
\else
   \usepackage[dvips]{graphicx}
   \DeclareGraphicsExtensions{.eps}
\fi
\usepackage{amsfonts}
\usepackage{amsthm}
\usepackage{algorithmic}
\usepackage{algorithm}
\usepackage[cmex10]{amsmath}
\usepackage{amssymb}
\usepackage{graphicx}
\usepackage{cite}
\usepackage{multirow}
\usepackage{color}

\newtheorem{proposition}{Proposition}

\begin{document}
%
\title{A Non-Alternating Algorithm for Joint BS-RS Precoding Design in Two-Way Relay Systems}

\author{Ehsan~Zandi, Guido~Dartmann,~\IEEEmembership{Student Member,~IEEE,}
	        Gerd~Ascheid,~\IEEEmembership{Senior Member,~IEEE},
	        Rudolf~Mathar,~\IEEEmembership{Member,~IEEE}
}
\maketitle
\begin{abstract}
Cooperative relay systems have become an active area of research during recent years since they help cellular networks to enhance data rate and coverage. In this paper we develop a method to jointly optimize
precoding matrices for amplify-and-forward relay station  and base station. Our objective is to increase max--min SINR fairness within co-channel users in a cell. The main achievement of this work is avoiding any tedious alternating
optimization for joint design of RS/BS precoders, in order to save complexity. 
Moreover, no convex solver is required in this method. RS precoding is done by transforming the underlying non-convex problem into a system of nonlinear equations which is
then solved using Levenberg-Marquardt algorithm. This method for RS precoder design is guaranteed to converge to a local optimum. For the BS precoder a low-complexity iterative method is proposed. The efficiency of the joint optimization method 
is verified by simulations.
\end{abstract}

\section{Introduction}

\IEEEPARstart{T}{his} paper considers a two-way relay system (TWRS) in which a multiple antenna relay station (RS) assists  a multiple antenna base station (BS) and several single antenna mobile stations (MS) to communicate. 
Such a system is expected to show a better performance in comparison with non-cooperative systems, in terms of both sum-rate and fairness among all users. Normally, the RS must be mounted sufficiently above roof top such that it can serve shadowed users.
Since the relay is not capable of simultaneous reception and transmission, the communication is split into two distinct half-duplex hops. In the first hop which is known as medium access (MAC), all users and the BS
transmit their data symbols towards the RS which amplifies and forwards the received signal in the next hop, known as broadcast (BC). The back propagated self-interference is known at each node and can be easily removed \cite{TaoTSPROC2012}
if all nodes have perfect knowledge of the channel coefficients.
The remaining concern is how to minimize the interference of other users. In this works we aim at maximizing the signal-to-interference-plus-noise ratio (SINR) of the weakest user link subject to some given power constraint. \\
\textbf{Related Work:}
The cooperative communication via an RS was first introduced in systems equipped with one-way relay channels where only one-way transmission (either downlink or uplink) is possible.
In \cite{HavaryTSPROC08}, the authors propose an optimal solution for a system consisting of one transmitting node and one receiving node via multiple one-way single antenna relay station. The bottleneck of 
this system compared to TWRS is that it requires four time hops for a complete bidirectional communication which leads to capacity loss. To overcome this shortcoming, TWRS are 
to reduce the number of hops to two and hence gain double throughput as the former case. Such a cooperative scenario for a single  pair of  users is studied in several works such as \cite{VazeITW2009,HavaryTSPROC2010,JingTSPROC2012,WangTCOM2012,ShahbazPanahiTSPROC2012}.
Later multiuser TWRS (MU-TWRS) with several transmitting pair of users, attracted a big attention as extension of the former one. 
So far, different scenarios are known as MU-TWRS. For instance \cite{SchadCAMSAP2011} and \cite{BournakaGC2011} consider multiple single antenna relays, while in \cite{TaoTSPROC2012}, \cite{ChenTWCCOM2009} and \cite{ZhangWSA2012} 
the system is equipped with one multiple antenna RS. One frequently used objective function in these works is to maximize equalized SINR among users \cite{TaoTSPROC2012}, \cite{SchadCAMSAP2011} and \cite{BournakaGC2011}. 
The corresponding optimization problem for finding the RS precoding matrix is known as max--min SINR optimization which is non-convex and $\mathcal{NP}$-hard in general \cite{GershmanMagazine10}. Most proposed algorithms are oriented around convexifying this problem
by means of various relaxation/approximation methods \cite{TaoTSPROC2012}. While the aforementioned works regard the so-called symmetric cooperative relay systems, many other works investigate the asymmetric case in which
a multiple antenna BS (and RS) and several single antenna users exist \cite{WanCommLett2013,WangTWCCOM2012,Sun2012,ZhangICASSP2011,Toh2009}.
None of these works offer any optimal solutions and come up with sub-optimal methods. In \cite{ZhangICASSP2011} and \cite{Toh2009} authors propose closed-form solutions based on zero-forcing (ZF). But, ZF is known
to impose constraint on the number of antennas on the system while attaining optimal solution only in very high signal-to-noise (SNR) regimes. On the other hand,  \cite{WanCommLett2013,WangTWCCOM2012,Sun2012} 
incorporate an alternating method into their algorithms which increases the computational complexity.

{\bf Contribution:} Most recent works, i.e., \cite{WanCommLett2013,WangTWCCOM2012,Sun2012}, handle the BS-RS precoding via alternating optimization. Alternating means that the problem is chopped into two subproblems, one for BS and another for RS precoder. One of the problems is solved while
the other precoder is  fixed. The solution is then used as fixed input for the other problem. This whole process continues iteratively to find a suboptimal solution. It is known that the max--min fairness problem for RS is non-convex 
\cite{TaoTSPROC2012}. Nevertheless, the optimization problem of BS precoder is convex, the overall alternating problem becomes non-convex and hence not convergent to the global optimum, in general. 
More importantly, the alternating scheme is very slow. To overcome this shortcoming, 
we isolate two subproblems and propose a non-alternating method which offers a suboptimal low-complexity solution to the joint design problem. Our method for BS precoder design requires no knowledge 
of the RS precoder since it only tries to maximize the upper bound on SINR of the latter. Then the resultant BS precoder is used as fixed matrix within the RS problem. The main benefit of the proposed method is that it induces
no outer (alternating) iterations on the design algorithm and a joint design is carried out in one iteration.
The RS precoder can be computed right away after BS precoder design at BS and be transmitted as side information along with data symbols. 
More importantly, most existing methods for BS/RS precoders utilize optimization tools which are in turn non real-time and prolong the process of joint optimization, exceedingly. 
In this paper, we convert RS precoder problem into a system of nonlinear equations and then solve it iteratively by Levenberg-Marquardt (LM) method. Elaborate explanation of this method can be found in \cite{DartmannTVT13} in which
convergence is proven. On the other hand, design of BS precoder is converted to a non-convex problem which is solved fast by introducing proper approximations. 
The rest of the paper is organized as follows. In Section \ref{sec:system_model} the system model and configuration is described. Section \ref{sec:Optimization_RS} formulates the optimization problem for
RS precoding design while Section \ref{sec:Optimization_BS} discusses the proposed algorithm for BS precoding problem. Numerical results along with complexity analysis are given in Section \ref{sec:simulation}.
Finally, Section \ref{sec:conclusion} concludes this paper.

{\bf Notation:} Along this paper, upper and lower case boldface symbols denote matrix and vector, respectively. The $i^\text{th}$ row (or column) of matrix $\mathbf{A}$ is shown by $[\mathbf{A}]_{i,:}~(\text{or}~[\mathbf{A}]_{:,i})$
while $[\mathbf{A}]_{i,j}$ refers to the element with indices $i$ and $j$. The identity matrix of size $m\times m$ is denoted by $\mathbf{I}_m$. Operators $(\;\cdot\;)^T$, and $(\;\cdot\;)^H$, $(\;\cdot\;)^*$ and $\otimes$ stands for
transpose, hermitian, complex conjugate and Kronecker product, while functions $\mathcal{E}(\;\cdot\;)$, $\mathcal{H}(\;\cdot\;)$ and $\operatorname{Tr}(\;\cdot\;)$ refer to expected value, harmonic mean and trace, respectively. 
Also magnitude of a complex number and the Euclidean norm of vector (or matrix) are shown by  $|\;\cdot\;|$, and $||\;\cdot\;||$. The vectorized version of a matrix $\mathbf{A}$ which is obtained by stacking its columns  into a single column and is denoted by $\operatorname{vec}(\mathbf{A})$. The $i^\text{th}$ eigenvalue and a vector of all eigenvalues of matrix $\mathbf{A}$ is shown by $\lambda_i(\mathbf{A})$ and $\boldsymbol{\lambda}(\mathbf{A})$, respectively.
\section{Data Model and System Setup} \label{sec:system_model}
This paper considers a system consisting of one BS equipped with $N_b$ antennas, one RS with $N_r$ antenna and $ N_u\geq 1$ users ($N_u \leq N_r$). For simplicity we choose
$N_u=N_b$ and assume there exits no LoS or specular link between BS and users. We also assume that channel state information (CSI) is perfectly known at each node. 
Figure \ref{system_fig} depicts the setting of the considered system including all channels and precoding matrices. In MAC phase the received signal at the RS is determined by
\begin{equation}
 \mathbf{r}_R=\mathbf{H}_2 \mathbf{W} \mathbf{x}_2+\mathbf{H}_1\mathbf{T}_1\mathbf{x}_1+\mathbf{n}_R.
\end{equation}
where $\mathbf{W}\in\mathbb{C}^{N_b\times N_u}$ is the BS precoding matrix. Throughout the paper we refer to user and BS with indices 1 and 2, respectively. $\mathbf{H}_2\in\mathbb{C}^{N_r\times N_b}$ and
$\mathbf{H}_1\in\mathbb{C}^{N_r\times N_u}$ denote the channel from BS to RS and users to RS. More precisely, the $i^\text{th}$ column of $\mathbf{H}_2$ represents the channel from  $i^\text{th}$
antenna element of BS towards the RS, while $[\mathbf{H}_1]_{i,j}$ corresponds to propagation channel between $j^\text{th}$ user and $i^\text{th}$ antenna of the RS. 
The matrix $\mathbf{T}_1=\operatorname{diag}(\sqrt{P_{1}}, \ldots, \sqrt{P_{N_u}})$ notifies the transmission power at MS side with $P_{i}$ being power constraint per MS. For simplicity this paper constrains
 all MSs to have the same power budget, i.e., $\mathbf{T}_1=\sqrt{\frac{P_U}{N_u}}\,\mathbf{I}_{N_u}$ where $P_U$ is a given sum-power constraint for the users.
We assume data symbols $\mathbf{x}_k \in\mathbb{C}^{N_u\times 1}$ at both ends are normalized, independent and identically distributed (i.i.d) such that 
$\mathcal{E}(\mathbf{x}_k\mathbf{x}^H_k)=\mathbf{I}_{N_u},\;k\in\{1,2\}$. It is necessary to mention that $\mathbf{x}_1=[x_1, \ldots, x_{N_u}]^T$ and $x_i$ is a transmitted symbol from $i^\text{th}$
MS. Furthermore, the sum power constraint at BS equals to $P_B$, that is  $\operatorname{Tr}(\mathbf{W}^H\mathbf{W})\leq P_B$. The additive noise $\mathbf{n}_R\sim \mathcal{CN}(0,\sigma^2_R\mathbf{I}_{N_r})$ at RS,
is assumed to be zero-mean circular complex Gaussian noise (CCGN), i.e.,
$\mathcal{E}(\mathbf{n}_R\mathbf{n}^H_R)=\sigma^2_R\mathbf{I}_{N_r}$. 
After performing amplify-and-forward precoding, the RS broadcasts the signal $\mathbf{s}_R=\boldsymbol{\Omega} \mathbf{r}_R$ in BC phase,
where $\boldsymbol{\Omega}\in\mathbb{C}^{N_r\times N_r}$ is the corresponding RS precoding matrix. The BS and the users receive the disturbed signals \eqref{downlink-complete}
and \eqref{uplink-complete}:
\begin{align}
\text{MS:}~(k=1) \hspace{.3cm} \mathbf{r}_1&=\mathbf{H}_1^T \mathbf{s}_R+\mathbf{n}_1 
=\underbrace{\mathbf{H}_1^T}_{\mathbf{A}_1}\boldsymbol{\Omega}\underbrace{\mathbf{H}_2\mathbf{W}}_{\mathbf{B}_1}\mathbf{x}_2+\mathbf{H}_1^T \boldsymbol{\Omega}\underbrace{\mathbf{H}_1\mathbf{T}_1}_{\mathbf{C}_1}\mathbf{x}_1+\mathbf{H}_1^T \boldsymbol{\Omega}\mathbf{n}_R+\mathbf{n}_1 \label{downlink-complete} \\ 
\text{BS:}~(k=2) \hspace{.3cm} \mathbf{r}_2&=\mathbf{H}_2^T \mathbf{s}_R+\mathbf{n}_2 
=\underbrace{\mathbf{H}_2^T}_{\mathbf{A}_2}\boldsymbol{\Omega}\underbrace{\mathbf{H}_1\mathbf{T}_1}_{\mathbf{B}_2}\mathbf{x}_1+\mathbf{H}_2^T \boldsymbol{\Omega}\underbrace{\mathbf{H}_2\mathbf{W}}_{\mathbf{C}_2}\mathbf{x}_2+\mathbf{H}_2^T \boldsymbol{\Omega}\mathbf{n}_R+\mathbf{n}_2. \label{uplink-complete}
 \end{align}
$\mathbf{n}_2\in\mathbb{C}^{N_b\times 1}$ describes the additive noise at BS antennas, while $\mathbf{n}_1\in\mathbb{C}^{N_u\times 1}$ is built by concatenating additive noise terms at each MS into a single vector. Both noise terms $\mathbf{n}_1, \mathbf{n}_2 \sim \mathcal{CN}(0,\sigma^2\mathbf{I}_{N_u})$ have zero-mean and variance $\sigma^2$, that is $\mathcal{E}(\mathbf{n}_1 \mathbf{n}^H_1)=\mathcal{E}(\mathbf{n}_2 \mathbf{n}^H_2)=\sigma^2 \mathbf{I}_{N_u}$.
In order to have a compact representation we rephrase \eqref{downlink-complete} and \eqref{uplink-complete} into the following form:
\begin{align}
\mathbf{r}_k&=\mathbf{A}_k \boldsymbol{\Omega} \mathbf{B}_k \mathbf{x}_{\overline{k}}+\mathbf{A}_k \boldsymbol{\Omega} \mathbf{C}_k \mathbf{x}_{k} 
+\mathbf{A}_k \boldsymbol{\Omega}\mathbf{n}_R+\mathbf{n}_k  ,\;\; k \in\{1,2\}, \overline{k} \neq k.
\end{align}
It is clear that terms including $\mathbf{B}_k$ include useful signal while terms with $\mathbf{C}_k$ represent interference. Note that the BS is aware of  
its back-propagated signal, i.e., $\mathbf{x}_2$ as well as all channels. Thus, the BS can remove this self-interference which leads to $\mathbf{C}_2=\boldsymbol{0}$. Similarly, each MS can remove its self-interference term which leads to equation \eqref{SINR1}  as the SINR for each MS stream.
\begin{align}\label{SINR1}
 \gamma_i^k&=\frac{|[\mathbf{A}_k\,\boldsymbol{\Omega}\,\mathbf{B}_k]_{i,i}|^2}{\sum_{\substack{j=1\\j\neq i}}^{N_u} |[\mathbf{A}_k\,\boldsymbol{\Omega}
\,\mathbf{B}_k]_{i,j}|^2+\sum_{\substack{j=1\\j\neq i}}^{N_u} |[\mathbf{A}_k\,\boldsymbol{\Omega}\,\mathbf{C}_k]_{i,j}|^2
+\sigma^2_R\,||[\mathbf{A}_k\, \boldsymbol{\Omega}]_{i,:}||^2+\sigma^2}.
\end{align}

\section{Optimization Problem} 
\subsection{Optimization at RS}  \label{sec:Optimization_RS}
First, we consider fixed $\mathbf{W}$ and optimize the max--min fairness problem over $\boldsymbol{\Omega}$. In the next section we will perform autonomous BS precoding regardless of the choice of $\boldsymbol{\Omega}$. For simplicity of notation, we merge uplink and downlink SINRs into one set, i.e.,
 $\{\gamma_1,\ldots,\gamma_{2N_u}\}=\{\gamma^1_1,\ldots\gamma^1_{N_u},\gamma^2_1,\ldots\gamma^2_{N_u}\},\;i\in \mathcal{A},\; j\in \mathcal{S}$,
and define sets $\mathcal{A}=\{{1,\ldots,2N_u}\}$ and $j\in \mathcal{S}=\{1,\ldots N_u\}$. Thus, the max--min fairness problem can be expressed by the following optimization problem:
\begin{align} \label{MMF-RS1}
\gamma^*=& \max_{\boldsymbol{\Omega}}\; \min_{i \in \mathcal{A}} \;\; \gamma_i \\
&\text{s.t.} \hspace{0.4cm} \operatorname{Tr}(\boldsymbol{\Omega}\,\mathbf{Y}\,\boldsymbol{\Omega}^H) \leq P_R, \nonumber
\end{align}
where $\mathbf{Y}\,=\mathbf{H}_2\,\mathbf{W}\,\mathbf{W}^H\,\mathbf{H}_2^H+\frac{P_U}{N_u}\mathbf{H}_1\,\mathbf{H}_1^H+\sigma^2\,\mathbf{I}_{N_r}$ and $P_R$ is the maximum allowed transmit power at RS. This problem is non-convex due to non-convexity of its objective function. 

Now, we rewrite \eqref{MMF-RS1} into a simpler equivalent form as in \cite{DartmannTVT13}. Let us define $\boldsymbol{\omega}=\operatorname{vec}(\boldsymbol{\Omega})$ and matrices
$\mathbf{N}^k_i=(\sigma^2_R\mathbf{I}_{N_r})\otimes([\mathbf{A}_k]_{i,:}^H[\mathbf{A}_k]_{i,:})$,
$\mathbf{Q}^k_{ij}=(\mathbf{q}^k_{ij})(\mathbf{q}^k_{ij})^H,\;\mathbf{q}^k_{ij}=\big([\mathbf{B}^k]_{:,j}^T\otimes[\mathbf{A}^k]_{i,:}\big)^H$, 
$\mathbf{S}^k_{ij}=(\mathbf{s}^k_{ij})(\mathbf{s}^k_{ij})^H,\;\mathbf{s}^k_{ij}=\big([\mathbf{C}^k]_{:,j}^T\otimes[\mathbf{A}^k]_{i,:}\big)^H$, 
$\mathbf{Q}^k_{i}=\mathbf{Q}^k_{ii}$ and 
$\mathbf{P}^k_{i}=\mathbf{N}^k_i+\sum_{\substack{j=1\\j\neq i}}^{N_u}\mathbf{S}^k_{ij}+\sum_{\substack{j=1\\j\neq i}}^{N_u}\mathbf{Q}^k_{ij},\;\; i,j \in \mathcal{S}$.
\newline
After removing indices $k$, i.e., $\{\mathbf{Q}_{1},\ldots,\mathbf{Q}_{2N_u}\}=\{\mathbf{Q}^1_{1},\ldots,\mathbf{Q}^1_{N_u},\mathbf{Q}^2_{1},\ldots,\mathbf{Q}^2_{N_u}\}$ and also similarly $\{\mathbf{P}_{1},\ldots,\mathbf{P}_{2N_u}\}=\{\mathbf{P}^1_{1},\ldots,\mathbf{P}^1_{N_u},\mathbf{P}^2_{1},\ldots,\mathbf{P}^2_{N_u}\}$, the objective function in \eqref{MMF-RS1} can be re-written as:
\begin{align} \label{MMF-RS2}
 \gamma^*&=\max_{\boldsymbol{\omega}}\;\;  \min_{\substack{i \in \mathcal{A}}}  
\frac{\boldsymbol{\omega}^H\mathbf{Q}_{i}\boldsymbol{\omega}}{\boldsymbol{\omega}^H\mathbf{P}_i\boldsymbol{\omega}+\sigma^2} \\
&\text{s.t.} \hspace{0.4cm} \boldsymbol{\omega}^H\,\mathbf{Z}\,\boldsymbol{\omega} \leq P_R, \nonumber
\end{align}
where $\mathbf{Z}=\mathbf{Y}^T\otimes\mathbf{I}_{N_r}$. For detailed explanation of the above transformation, we encourage readers to see \cite{DartmannTVT13}. The objective function \eqref{MMF-RS2} is non-convex and $\mathcal{NP}$-hard in general, since it is a quadratic fractional program and makes  
our problem non-convex \cite{Palomar10}. This problem can be solved using semidefinite relaxation (SDR) \cite{Bengtsson02} using the fact 
$\boldsymbol{\omega}^H\mathbf{Q}_{i}\boldsymbol{\omega}=\operatorname{Tr}(\mathbf{Q}_{i}\boldsymbol{\omega}\boldsymbol{\omega}^H)$. Let $\mathbf{X}=\boldsymbol{\omega}\boldsymbol{\omega}^H$ and assume $\operatorname{rank}(\mathbf{X})=1$, then \eqref{MMF-RS2} can be solved by bisection over $\gamma$ as  a feasbility check problem. For each value of $\gamma$ the following semidefinite program (SDP) must be solved:
\vspace{-0.2cm}
\begin{align} \label{MMF-SDP-RS}
& \operatorname{find}\;\; \mathbf{X}  \\
\vspace{-0.2cm}
&\text{s.t.} \hspace{0.4cm} \operatorname{Tr}\big(\mathbf{Z}\mathbf{X}\big) \leq P_R,
\hspace{0.2cm} \operatorname{Tr}\big((\frac{1}{\gamma}\mathbf{Q}_{i}-\mathbf{P}_i)\mathbf{X}\big) \geq \sigma^2,\hspace{0.2cm} \mathbf{X} \succeq 0, \nonumber
\end{align}
where $\mathbf{X} \succeq 0$ means that $\mathbf{X}$ is positive semidefinite. 
If $\operatorname{rank}(\mathbf{X})=1$ the problem \eqref{MMF-SDP-RS} is solved optimally, 
otherwise other techniques such as randomization \cite{LuoMagazin10} and dominant eigenvector decomposition  can be applied \cite{KARIPIDIS08}. 
Now, we solve the problem \eqref{MMF-RS2} with a low-complex method presented in \cite{DartmannTVT13} which 
takes advantage of minimax inequality to transform \eqref{MMF-RS2} into a system of non-linear equations and can be interpreted as a least-squares problem. In \cite{DartmannTVT13}, the system of 
equations is solved by LM method in which bisection over SINR is utilized to comply with the power constraint. In \cite{DartmannTVT13}, it is proved that
problem \eqref{MMF-RS2} is upper-bounded by \eqref{minimax_upper_bound}. This bound is used in the next section for BS precoding design.
\begin{proposition} \cite{DartmannTVT13} \label{proposition_minimax_upper_bound}
Let $\tilde{\mathbf{P}}_i=\mathbf{P}_i+\frac{\sigma^2}{P_R}\mathbf{Z}$, then problem \eqref{MMF-RS2} is upper bounded by 
\begin{align} \label{minimax_upper_bound}
\hat{\gamma} = \min_{i \in \mathcal{A}} \;\; \lambda_{\text{max}} \big(\tilde{\mathbf{P}}_i^{-1} \mathbf{Q}_i\big)=\min_{i \in \mathcal{A}} \;\; 
\lambda_{\text{max}} \big(\tilde{\mathbf{P}}_i^{-1/2} \mathbf{Q}_i \tilde{\mathbf{P}}_i^{-1/2}\big).
\end{align}
\end{proposition}
\begin{proof}
This proposition can be easily proved based on results in \cite{DartmannTVT13}.
Hereafter, we refer to $\hat{\gamma}$ as minimax bound of probelm \eqref{MMF-RS2}. 
\end{proof}

\subsection{Optimizing the BS precoder} \label{sec:Optimization_BS}
Before introducing our proposed method, we need to formulate BS precoding problem similar to \eqref{MMF-SDP-RS} to be solved by SDP-bisection technique (for a given $\boldsymbol{\Omega}$) within
alternating method. Let $\mathbf{F}=\mathbf{H}^T_1\,\boldsymbol{\Omega}\,\mathbf{H}_2$, $\boldsymbol{\theta}=\operatorname{vec}(\mathbf{W})$, 
$\sigma^2_{w,i}=\big((\sum_{\substack{j=1\\j\neq i}}^{N_u} |[\mathbf{H}^T_1\,\boldsymbol{\Omega}\,\mathbf{H}_1\mathbf{T}_1]_{i,j}|^2)+\sigma^2_R\,||[\mathbf{H}^T_1\,\boldsymbol{\Omega}]_{i,:}||^2+\sigma^2\big)$,
$\tilde{\mathbf{I}}_j=([\mathbf{I_{N_u}}]_{j,:})^T[\mathbf{I_{N_u}}]_{j,:}$, $\tilde{\mathbf{F}}_i=([\mathbf{F}]_{i,:})^H[\mathbf{F}]_{i,:}$ and $\mathbf{X}=\boldsymbol{\theta}\boldsymbol{\theta}^H$.
Then, after similar operations as in Section \ref{sec:Optimization_RS}, the BS precoding design can be done similar to \eqref{MMF-SDP-RS}. Note that in the current problem $\mathbf{Z}$, $\mathbf{P}_i$ and $\mathbf{Q}_i$ are replaced with $\mathbf{I}_{N_b}$,
$\mathbf{D}_{i}=\sum_{\substack{j=1\\j\neq i}}^{N_u}\tilde{\mathbf{I}}_j \otimes \tilde{\mathbf{F}}_i$ and $\mathbf{C}_{i}=\tilde{\mathbf{I}}_i \otimes \tilde{\mathbf{F}}_i$, respectively. 
Also the power constraint in this case is $P_B$ and $\sigma^2$ must be replaced with $\sigma^2_{w,i}$ as well.
Now, we start to develop our proposed method. Let $[\mathbf{A}_1]_{i,:}=\mathbf{a}_i$ and $\mathbf{b}_j=([\mathbf{B}_1]_{:,j})^T=([\mathbf{H}_2\mathbf{W}]_{:,j})^T=(\mathbf{H}_2\mathbf{w}_j)^T$, where $\mathbf{w}_j=[\mathbf{W}]_{:,j}
,\; i,j \in \mathcal{S}$, then 
$\mathbf{q}^1_{ij}=\mathbf{b}_j^H \otimes \mathbf{a}^H_i$ which results in $\mathbf{Q}^1_{ij}=(\mathbf{b}_j^H\mathbf{b}_j) \otimes (\mathbf{a}^H_i\mathbf{a}_i)$. 
According to Proposition \ref{proposition_minimax_upper_bound}, problem \eqref{MMF-RS2} is upper bounded by \eqref{minimax_upper_bound}, so we aim at maximizing $\min_{i \in \mathcal{A}} \; \lambda_{\text{max}} \big(\tilde{\mathbf{P}}_i^{-1} \mathbf{Q}_i\big)$.
Roots of the characteristic polynomial of $\tilde{\mathbf{P}}_i^{-1} \mathbf{Q}_i$ and $\tilde{\mathbf{P}}_i^{-1/2} \mathbf{Q}_i \tilde{\mathbf{P}}_i^{-1/2}$ for a given $i$ are identical 
since they are similar matrices \cite{Meyer2000}. 
Hence, the first one inherits the property of positive semidefiniteness from the latter. 
Notice that $\mathbf{Y}$ (and hence $\mathbf{Z}$) is  the only relevant term in determining uplink SINR which is affected by $\mathbf{W}$. Thus, we
expect negligible influence of BS precoding on uplink SINR. This is also evidenced by simulations. Also, unlike downlink, the back propagated 
interference from other users can be removed in the uplink. Consequently, the maximum eigenvalues of corresponding matrices  in downlink are likely to
be smaller than those of the uplink. Taking this reasoning into consideration, we hope that maximizing the corresponding maximum eigenvalues in downlink
leads to an increase in minimax bound. Let $\mathcal{S}'=\{1,\ldots,N_u\}$ denote  all indices relevant to downlink,
e.g., $\mathbf{Q}_i,\;i \in \mathcal{S}'$ refers to  $\mathbf{Q}^1_i,\;i \in \mathcal{S}$.  We know that $\mathbf{Q}_i,~\mathbf{P}_i$ and 
$\mathbf{b}_i$ are functions of $\mathbf{W}$, but for simpler representation we omit to write them as $ \mathbf{Q}_i(\mathbf{W})$, $ \mathbf{P}_i(\mathbf{W})$ and  $\mathbf{b}_i(\mathbf{W})$.
From now on, we presume all indices $i$ belong to $\mathcal{S}'$. The following optimization problem is non-convex:
\begin{align} \label{MMF-BS1}
& \max_{\mathbf{W}}\; \min_{i \in \mathcal{S}'} \;\; \lambda_{\text{max}} \big(\tilde{\mathbf{P}}_i^{-1/2} \mathbf{Q}_i \tilde{\mathbf{P}}_i^{-1/2}\big) \\
&\text{s.t.} \hspace{0.4cm} \operatorname{Tr}(\mathbf{W}^H\mathbf{W}) \leq P_B. \nonumber
\end{align}
But $\lambda_{\text{max}} \big(\tilde{\mathbf{P}}_i^{-1/2} \mathbf{Q}_i \tilde{\mathbf{P}}_i^{-1/2}\big)$ is in turn bounded by \cite{Meyer2000}:
\begin{equation}
\label{norm_bound_lambda}
\lambda_{\text{max}} \big(\tilde{\mathbf{P}}_i^{-1/2} \mathbf{Q}_i \tilde{\mathbf{P}}_i^{-1/2}\big) \leq ||\tilde{\mathbf{P}}_i^{-1/2} \mathbf{Q}_i \tilde{\mathbf{P}}_i^{-1/2}||
\leq  ||\mathbf{Q}_i||\,||\tilde{\mathbf{P}}_i^{-1/2}||^2.
\end{equation}
Using the fact that $(\mathbf{X}_1 \otimes \mathbf{Y}_1)(\mathbf{X}_2\otimes \mathbf{Y}_2)=(\mathbf{X}_1 \mathbf{X}_2)\otimes (\mathbf{Y}_1 \mathbf{Y}_2)$ and also
 $\operatorname{Tr}(\mathbf{X}\otimes\mathbf{Y})=\operatorname{Tr}(\mathbf{X})\operatorname{Tr}(\mathbf{Y})$, it is easy to see that $||\mathbf{Q}_i||=||\mathbf{a}_i||^2\,||\mathbf{b}_i||^2$ 
 and
\begin{align} \label{bound_P-tilde}
||\tilde{\mathbf{P}}_i^{-1/2}||^2&=\operatorname{Tr}\big((\tilde{\mathbf{P}}_i^{-1/2})^H\tilde{\mathbf{P}}_i^{-1/2}\big)=\operatorname{Tr}\big(\tilde{\mathbf{P}}_i^{-1}\big)
=\sum_{k=1}^{N^2_r} \lambda_k(\tilde{\mathbf{P}}_i^{-1})
=\frac{N^2_r}{\mathcal{H}\big(\boldsymbol{\lambda}(\tilde{\mathbf{P}}_i)\big)}.
\end{align}
Since $0 < \lambda_\text{min}(\tilde{\mathbf{P}}_i) \leq  \mathcal{H}\big(\boldsymbol{\lambda}\big(\tilde{\mathbf{P}}_i)\big) \leq \lambda_\text{max}(\tilde{\mathbf{P}}_i)$, $\mathcal{H}\big(\boldsymbol{\lambda}(\tilde{\mathbf{P}}_i )\big)$ 
may decrease (and $||\tilde{\mathbf{P}}_i^{-1/2}||^2$ increases accordingly) by decreasing $\lambda_\text{max}(\tilde{\mathbf{P}}_i)$. 
Thus, we proceed by decreasing the bound on $\lambda_\text{max}(\tilde{\mathbf{P}}_i)$:
\begin{align}
\lambda_\text{max}(\tilde{\mathbf{P}}_i) \leq ||\tilde{\mathbf{P}}_i||=||\frac{\sigma^2}{P_R}\mathbf{Z}+\mathbf{N}_i+\sum_{\substack{j=1\\j\neq i}}^{N_u}\mathbf{S}_{i,j}+\sum_{\substack{j=1\\j\neq i}}^{N_u}\mathbf{Q}_{i,j}||
,\;\; i \in \mathcal{S}'.
\end{align}
which can be further expanded using triangle inequality and the fact that $||\mathbf{Q}_{ij}||=||\mathbf{a}_i||^2\,||\mathbf{b}_j||^2$:
\begin{align} \label{lambda_max_P-tilde}
\lambda_\text{max}(\tilde{\mathbf{P}}_i) \leq ||\frac{\sigma^2}{P_R}\mathbf{Z}+\mathbf{N}_i+\sum_{\substack{j=1\\j\neq i}}^{N_u}\mathbf{S}_{i,j}||+\sum_{\substack{j=1\\j\neq i}}^{N_u}||\mathbf{Q}_{i,j}|| 
=||\frac{\sigma^2}{P_R}\mathbf{Z}+\mathbf{N}_i+\sum_{\substack{j=1\\j\neq i}}^{N_u}\mathbf{S}_{i,j}||+||\mathbf{a}_i||^2\,\sum_{\substack{j=1\\j\neq i}}^{N_u}||\mathbf{b}_j||^2.
\end{align}
Putting \eqref{norm_bound_lambda}-\eqref{lambda_max_P-tilde} all together we conclude that our optimization problem is:
\begin{align} \label{MMF-BS2}
& \max_{\mathbf{W}}\; \min_{i \in \mathcal{S}'} \;\; \frac{||\mathbf{b}_i||}{\sum_{\substack{j=1\\j\neq i}}^{N_u}||\mathbf{b}_j||} \\
&\text{s.t.} \hspace{0.4cm} \operatorname{Tr}(\mathbf{W}^H\mathbf{W}) \leq P_B. \nonumber
\end{align}
Remember that each $\mathbf{b}_i$ is a function of $\mathbf{W}$. It is not straightforward to determine the convexity of this problem. However, for each user we 
need to maximize $||\mathbf{b}_i||$ while minimizing $||\mathbf{b}_j||,\forall j \neq i$. However, this decreases the upper bound for other users, as an outcome.
Hence, the best case is to equalize $||\mathbf{b}_i||$ for all users to satisfy \eqref{MMF-BS2}:
\begin{align} 
||\mathbf{b}_i||&=||(\mathbf{H}_2\mathbf{w}_i)^T||=\sqrt{\mathbf{w}_i^H\mathbf{H}_2^H\mathbf{H}_2\mathbf{w}_i}=\text{const} ,\;\ \forall i \in \mathcal{S}'.
\end{align}
Since $\mathbf{H}_2^H\mathbf{H}_2$ is Hermitian and thus normal, it is unitarily diagonalizable \cite{Tse05}, i.e.,
$\mathbf{H}_2^H\mathbf{H}_2=\mathbf{U}\mathbf{\Lambda}\mathbf{U}^H ,\;\mathbf{U}^{-1}=\mathbf{U}^H$, 
where $\mathbf{\Lambda}$ is a diagonal matrix whose diagonal entries are eigenvalues of $\mathbf{H}_2^H\mathbf{H}_2$. Let $ \boldsymbol{\Psi}=\operatorname{diag}(\psi_1,\ldots,\psi_{N_u})$ and 
$\mathbf{H}_0=\mathbf{U}\mathbf{\Lambda}^{-1/2}$, then we choose 
$\mathbf{W}\big(\boldsymbol{\Psi}\big)=\mathbf{H}_0\boldsymbol{\Psi}$. Therefore,
$\mathbf{W}^H\mathbf{H}_2^H\mathbf{H}_2\mathbf{W}=\boldsymbol{\Psi}^2$ and also $||\mathbf{b}_i||=\psi_i$. 
The positive real-valued matrix $\boldsymbol{\Psi}$ relates to power allocation at BS, i.e., 
$\operatorname{Tr}(\mathbf{W}^H\mathbf{W})=\operatorname{Tr}(\boldsymbol{\Psi}^2\mathbf{\Lambda}^{-1}) \leq P_B$.
If $\boldsymbol{\Psi}=\mathbf{I}_{N_b}$ then \eqref{MMF-BS2} is solved but not \eqref{MMF-BS1}, since  $\lambda_{\text{max}} \big(\tilde{\mathbf{P}}_i^{-1} \mathbf{Q}_i \big)$  are not equal for all users. In order to
equalize all $\lambda_{\text{max}} \big(\tilde{\mathbf{P}}_i^{-1} \mathbf{Q}_i \big)$s, we now propose a fast convergent method for power allocation at BS.
Let $l$ represent the iteration index. Apparently, positive definite matrices $\mathbf{X}_i(l)=\tilde{\mathbf{P}}_i^{-1}(l)+\mathbf{Q}_i(l)$ change from iteration to iteration 
by varying power distribution  amongst users. It is easy to see:
\begin{align} 
\lambda_\text{max}\big(\tilde{\mathbf{P}}_i^{-1}(l)\mathbf{Q}_i(l)\big)&=
\lambda_\text{max}\big(\tilde{\mathbf{P}}_i^{-1}(l)\mathbf{X}_i(l)\big)-1=
\frac{1}{\lambda_\text{min}\big(\tilde{\mathbf{X}}_i^{-1}(l)\mathbf{P}_i(l)\big)}-1 \nonumber\\=
&\frac{1}{1-\lambda_\text{max}\big(\tilde{\mathbf{X}}_i^{-1}(l)\mathbf{Q}_i(l)\big)}-1=
\frac{\lambda_\text{max}\big(\mathbf{X}_i^{-1}(l)\mathbf{Q}_i(l)\big)}{1-\lambda_\text{max}\big(\mathbf{X}_i^{-1}(l)\mathbf{Q}_i(l)\big)}.
\end{align}
We know that rank-1 matrix $\tilde{\mathbf{P}}_i^{-1}(l)\mathbf{Q}_i(l) \succeq 0$, then $\lambda_\text{max} \big(\tilde{\mathbf{P}}_i^{-1}(l)\mathbf{Q}_i(l)\big) \geq 0$. The same holds for $\mathbf{X}_i^{-1}(l)\mathbf{Q}_i(l)$, therefore it can be concluded that $0 \leq \lambda_\text{max}\big(\mathbf{X}_i^{-1}(l)\mathbf{Q}_i(l)\big) < 1$. 
The relation between these eigenvalues is represented in Figure \ref{fig:lambda}. Qualitatively speaking, in each iteration we need to decrease the power of users who have higher SINR and do the opposite for users have lower SINR. By doing so, after some iterations all users are expected to have the same SINR. It appears that if we find $\boldsymbol{\Psi}(l)$ such that equalizes $\mathbf{X}_i^{-1}(l)\mathbf{Q}_i(l)$, we have also equalized the upper bound on SINR for all users in the downlink case. Let $\psi_i(l)=\lambda^{-1}_\text{max}\big(\mathbf{X}_i^{-1}(l)\mathbf{Q}_i(l)\big)$ be $i^\text{th}$ diagonal entry of  $\Psi(l)$ and $\mathbf{W}(l)$ be the BS precoder in the $l^\text{th}$ iteration, then
\begin{align} \label{Wt_closed_form}
\mathbf{W}(l)=\sqrt{\frac{P_B}{\operatorname{Tr}\big((\mathbf{W}'(l))^H\mathbf{W}'(l)\big)}}\;\mathbf{W}'(l) ,\; \mathbf{W}'(l)=\mathbf{W}(l-1)\boldsymbol{\Psi}(l-1).
\end{align}
It is significant to notice that the choice of $\mathbf{H}_0=\mathbf{U}\mathbf{\Lambda}^{-1/2}$ is rather heuristic and does not guarantee to increase the SINR for any given channels since it only aims at increasing the minimax upper bound on $\lambda_{\text{max}} \big(\tilde{\mathbf{P}}_i^{-1} \mathbf{Q}_i\big)$. 
Some cases may still exist, depending on the channel, in which BS precoding may either worsen SINR or achieve no benefit. In
such cases the BS precoding reduces to a simple power allocation of \eqref{Wt_closed_form} with $\mathbf{H}_0=\mathbf{I}_{N_b}$. 
We examine this with the harmonic mean of maximum eigenvalues of all downlink users. Algorithm \ref{Wt_algorithm} summarizes our proposed method for finding a sub-optimal solution for \eqref{MMF-BS1}. The resulting $\mathbf{W}$
is then used as fixed input for LM algorithm to find near-optimal solution of \eqref{MMF-RS2}.
\section{Simulation Results and Complexity Analysis} \label{sec:simulation}
To justify the proposed method, we have performed simulations for 1000 different realizations of channel matrices. The presented results are averaged over all attained data. We have assumed that 
$N_r=6,~ N_u=N_b=3,~ P_R=P_B=P_U=10^{1.5}$ and fixed $\sigma_R=1$ while varying $\sigma$ from $0.01$ to $1.5$. 
The underlying channel is Rayleigh fading which is generated similar to \cite{Narasimhan06}. First, channel matrices $\mathbf{H}_1$ and $\mathbf{H}_2$, whose entries are i.i.d Gaussian random variables with zero means and unit variances, are generated. Then, in order to create correlation between different links they are manipulated by
$\mathbf{H}_1=\mathbf{R}_{RS}^{1/2}\,\mathbf{H}_1\,\mathbf{R}_{MS}^{1/2}$ and $\mathbf{H}_2=\mathbf{R}_{RS}^{1/2}\,\mathbf{H}_2\,\mathbf{R}_{BS}^{1/2}$
where $[\mathbf{R}_{BS}]_{ij}=(\rho_{BS})^{|i-j|}$, $[\mathbf{R}_{RS}]_{ij}=(\rho_{RS})^{|i-j|}$ and $[\mathbf{R}_{MS}]_{ij}=(\rho_{MS})^{|i-j|}$. We have chosen values 
$\rho_{BS}=0.6172$, $\rho_{RS}=0.5883$ and $\rho_{MS}=0.1$.
\newline
We have solved the RS precoding optimization \eqref{MMF-RS2} using LM-bisection and also by SDP-bisection.
These are performed twice, i.e., first before BS precoding and then after precoding at BS using Algorithm \ref{Wt_algorithm}. Also as a reference the alternating method is simulated with number of iterations
of $N_a=6$. Even this method does not result in the optimal solution to the joint optimization method since one of the problems, i.e., RS precoding, is non-convex and so is the whole problem.
At the end we choose the best achieved rate amongst $N_a$ iterations regardless of which iteration yields so. 
\newline
Figure \ref{rate_fig} shows the achievable rate for the worst user in different methods along with minimax upper bound of \eqref{minimax_upper_bound}. Even though there exits no guarantee to achieve upper bound  in general, we observe in the figure that this bound is really tight in high SNR situations. Simulations admit that our proposed methods works well by handing in a relatively close rate to the alternating algorithm. Anyhow, the optimal solution is hard to find due to non-convexity of joint optimization problem. 
Figure \ref{iter_Wt_fig} illustrates that the average number of iterations for finding $\mathbf{W}$ is between $7-10$ depending on SNR value. This well confirms that Algorithm \ref{Wt_algorithm} is a fast method. Similarly, Figure \ref{LM_iter_fig} portrays the number of iterations in LM-bisection method. With an identical number of iterations, LM-bisection after BS precoding has a bigger gap with minimax bound compared to the case without precoding. Most probably, the reason is that after performing BS precoding the initial precoder (in Levenberg-Marquardt method), that we acquire from Proposition \ref{proposition_minimax_upper_bound},  is further away from the optimal beamformer in comparison with the case $\mathbf{W}=\mathbf{I}_{N_b}$.
\newline
Table \ref{complexity_table} compares the complexity of the proposed with the alternating method. In the table, $N_w$ and $N_{LM}$ refer to the number of iterations in Algorithm \ref{Wt_algorithm} and LM-bisection, respectively. Also in our simulation $\epsilon_{SDP}=1.489 \cdot 10^{-8}$ is used.
Figure \ref{rate_fig} along with Table \ref{complexity_table} show that the proposed method offers a good compromise between achieved worst-case rate and complexity. 
In the worst case (in the strongest noise conditions) the achieved rate of our method ($0.66~bits/sec/Hz$) amounts to $78\%$ of the latter. 

\section{Conclusion} \label{sec:conclusion}
This paper proposes a novel non-alternating low-complexity method for joint optimization of BS-RS precoders design. To find the near-optimal beamformer of RS, the  minimax upper bound 
is first found whose corresponding beamformer is then used as starting point of iterative Levenberg-Marquardt algorithm. The main idea of this algorithm as shown in \cite{DartmannTVT13} is to avoid any convex 
solver after converting the original non-convex problem to a set of nonlinear equations. The solution always converges to a local optimum. Later a fast iterative method is presented to find a sub-optimal solution for BS precoder. 
Indeed, two subproblems are isolated at the expense of small loss of rate in comparison with alternating approach. Simulations show that Algorithm \ref{Wt_algorithm}
always converges fast.

\section{Acknowledgement}
Authors would like to thank Dipl.-Ing. Gholamreza Alirezaei for his kind support and technical advice during this research.

\bibliographystyle{IEEEtran}
\bibliography{arxiv}

\begin{thebibliography}{10}
\providecommand{\url}[1]{#1}
\csname url@samestyle\endcsname
\providecommand{\newblock}{\relax}
\providecommand{\bibinfo}[2]{#2}
\providecommand{\BIBentrySTDinterwordspacing}{\spaceskip=0pt\relax}
\providecommand{\BIBentryALTinterwordstretchfactor}{4}
\providecommand{\BIBentryALTinterwordspacing}{\spaceskip=\fontdimen2\font plus
\BIBentryALTinterwordstretchfactor\fontdimen3\font minus
  \fontdimen4\font\relax}
\providecommand{\BIBforeignlanguage}[2]{{%
\expandafter\ifx\csname l@#1\endcsname\relax
\typeout{** WARNING: IEEEtran.bst: No hyphenation pattern has been}%
\typeout{** loaded for the language `#1'. Using the pattern for}%
\typeout{** the default language instead.}%
\else
\language=\csname l@#1\endcsname
\fi
#2}}
\providecommand{\BIBdecl}{\relax}
\BIBdecl

\bibitem{TaoTSPROC2012}
M.~Tao and R.~Wang, ``Linear precoding for multi-pair two-way {MIMO} relay
  systems with max-min fairness,'' \emph{{IEEE Transactions on Signal
  Processing}}, vol.~60, no.~10, pp. 5361--5370, Oct. 2012.

\bibitem{HavaryTSPROC08}
V.~Havary-Nassab, S.~Shahbazpanahi, A.~Grami, and Z.-Q. Luo, ``Distributed
  beamforming for relay networks based on second-order statistics of the
  channel state information,'' \emph{{IEEE Transactions on Signal Processing}},
  vol.~56, no.~9, pp. 4306--4316, Sep. 2008.

\bibitem{VazeITW2009}
R.~Vaze and R.~Heath, ``Optimal amplify and forward strategy for two-way relay
  channel with multiple relays,'' in \emph{IEEE Information Theory Workshop on
  Networking and Information Theory, (ITW 2009)}, Jun. 2009, pp. 181--185.

\bibitem{HavaryTSPROC2010}
V.~Havary-Nassab, S.~ShahbazPanahi, and A.~Grami, ``Optimal distributed
  beamforming for two-way relay networks,'' \emph{{IEEE Transactions on Signal
  Processing}}, vol.~58, no.~3, pp. 1238 --1250, Mar. 2010.

\bibitem{JingTSPROC2012}
Y.~Jing and S.~ShahbazPanahi, ``Max--min optimal joint power control and
  distributed beamforming for two-way relay networks under per-node power
  constraints,'' \emph{{IEEE Transactions on Signal Processing}}, vol.~60,
  no.~12, pp. 6576 --6589, Dec. 2012.

\bibitem{WangTCOM2012}
W.~Wang, S.~Jin, and F.~C. Zheng, ``Maximin {SNR} beamforming strategies for
  two-way relay channels,'' \emph{IEEE Communications Letters}, vol.~16, no.~7,
  pp. 1006 --1009, Jul. 2012.

\bibitem{ShahbazPanahiTSPROC2012}
S.~Shahbazpanahi and M.~Dong, ``A semi-closed-form solution to optimal
  distributed beamforming for two-way relay networks,'' \emph{{IEEE
  Transactions on Signal Processing}}, vol.~60, no.~3, pp. 1511 --1516, Mar.
  2012.

\bibitem{SchadCAMSAP2011}
A.~Schad and M.~Pesavento, ``Multiuser bi-directional communications in
  cooperative relay networks,'' in \emph{4th IEEE International Workshop on
  Computational Advances in Multi-Sensor Adaptive Processing (CAMSAP)}, Dec.
  2011, pp. 217--220.

\bibitem{BournakaGC2011}
G.~Bournaka, K.~Cumanan, S.~Lambotharan, and F.~Lazarakis, ``An iterative
  semidefinite and geometric programming technique for the {SINR} balancing in
  two-way relay network,'' in \emph{IEEE Global Telecommunications Conference
  (GLOBECOM 2011)}, Dec. 2011.

\bibitem{ChenTWCCOM2009}
M.~Chen and A.~Yener, ``Multiuser two-way relaying: detection and interference
  management strategies,'' \emph{{IEEE Transactions on Wireless
  Communications}}, vol.~8, no.~8, Aug. 2009.

\bibitem{ZhangWSA2012}
J.~Zhang, N.~Bornhorst, F.~Roemer, M.~Haardt, and M.~Pesavento, ``Optimal and
  suboptimal beamforming for multi-operator two-way relaying with a {MIMO}
  amplify-and-forward relay,'' in \emph{Int. ITG Workshop on Smart Antennas
  (WSA)}, Mar. 2012.

\bibitem{GershmanMagazine10}
A.~Gershman, N.~Sidiropoulos, S.~Shahbazpanahi, M.~Bengtsson, and B.~Ottersten,
  ``Convex optimization-based beamforming,'' \emph{Signal Processing Magazine,
  IEEE}, vol.~27, no.~3, pp. 62 --75, May 2010.

\bibitem{WanCommLett2013}
H.~Wan, W.~Chen, and W.~Xialoi, ``Joint source and relay design for {MIMO}
  relaying broadcast channels,'' \emph{Communications Letters, IEEE}, vol.~17,
  no.~2, pp. 345--348, Feb. 2013.

\bibitem{WangTWCCOM2012}
R.~Wang, M.~Tao, and Y.~Huang, ``Linear precoding designs for
  amplify-and-forward multiuser two-way relay systems,'' \emph{{IEEE
  Transactions on Wireless Communications}}, vol.~11, no.~12, pp. 4457 --4469,
  Dec. 2012.

\bibitem{Sun2012}
C.~Sun, C.~Yang, Y.~Li, and B.~Vucetic, ``Transceiver design for multi-user
  multi-antenna two-way relay cellular systems,'' \emph{Communications, IEEE
  Transactions on}, vol.~60, no.~10, pp. 2893--2903, Oct. 2012.

\bibitem{ZhangICASSP2011}
J.~Zhang, F.~Roemer, and M.~Haardt, ``Beamforming design for multi-user two-way
  relaying with {MIMO} amplify and forward relays,'' in \emph{Acoustics, Speech
  and Signal Processing (ICASSP), 2011 IEEE International Conference on}, May
  2011, pp. 2824--2827.

\bibitem{Toh2009}
S.~Toh and D.~T.~M. Slock, ``A linear beamforming scheme for multi-user {MIMO}
  af two-phase two-way relaying,'' in \emph{Personal, Indoor and Mobile Radio
  Communications, 2009 IEEE 20th International Symposium on}, 2009, pp.
  1003--1007.

\bibitem{DartmannTVT13}
G.~Dartmann, E.~Zandi, and G.~Ascheid, ``A modified {L}evenberg-{M}arquardt
  method for the bidirectional relay channel,'' \emph{submitted into IEEE
  Transactions on Vehicular Technology, available on
  http://arxiv.org/abs/1307.3121}, Jul. 2013.

\bibitem{Palomar10}
D.~Palomar and Y.~C. Eldar, \emph{Convex optimization in signal processing and
  communications}.\hskip 1em plus 0.5em minus 0.4em\relax Cambridge University
  Press, Dec. 2009.

\bibitem{Bengtsson02}
M.~Bengtsson and B.~Ottersten, ``Optimum and suboptimum transmit beamforming,''
  in \emph{Handbook of Antennas in Wireless Communications}.\hskip 1em plus
  0.5em minus 0.4em\relax Boca Raton, USA: CRC Press, 2002.

\bibitem{LuoMagazin10}
Z.-Q. Luo, W.-K. Ma, A.~M.-C. So, Y.~Ye, and S.~Zhang, ``Semidefinite
  relaxation of quadratic optimization problems,'' \emph{IEEE Signal Processing
  Magazine, IEEE}, vol.~27, no.~3, pp. 20--34, May 2010.

\bibitem{KARIPIDIS08}
E.~Karipidis, N.~D. Sidiropoulos, and Z.-Q. Luo, ``Quality of service and
  max--min fair transmit beamforming to multiple cochannel multicast groups,''
  \emph{{IEEE Transactions on Signal Processing}}, vol.~56, no.~3, pp.
  1268--1279, March 2008.

\bibitem{Meyer2000}
C.~D. Meyer, \emph{Matrix analysis and applied linear algebra}.\hskip 1em plus
  0.5em minus 0.4em\relax Philadelphia, PA, USA: Society for Industrial and
  Applied Mathematics, 2000.

\bibitem{Tse05}
D.~N.~C. Tse and P.~Viswanath, \emph{Fundamentals of wireless
  communication}.\hskip 1em plus 0.5em minus 0.4em\relax Cambridge University
  Press, 2005.

\bibitem{Narasimhan06}
R.~Narasimhan, ``Finite-{SNR} diversity-multiplexing tradeoff for correlated
  {R}ayleigh and {R}ician {MIMO} channels,'' \emph{{IEEE Transactions on
  Information Theory}}, vol.~52, no.~9, pp. 3965--3979, Sep. 2006.

\bibitem{SIDIRO2006}
N.~D. Sidiropoulos, T.~N. Davidson, and Z.-Q. Luo, ``Transmit beamforming for
  physical-layer multicasting,'' \emph{{IEEE Transactions on Signal
  Processing}}, vol.~54, no.~6, pp. 2239--2251, Jun. 2006.

\end{thebibliography}

\newpage

 \begin{figure}[t]
 \begin{center}
     \includegraphics[width=\columnwidth]{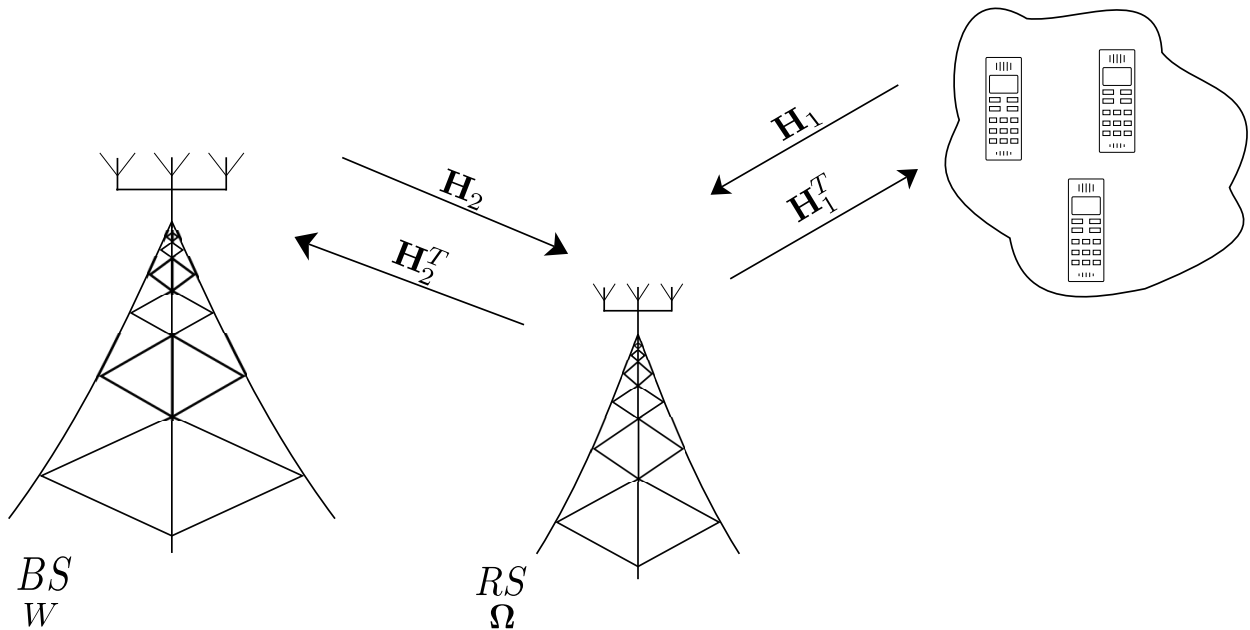}
     \caption{\small System setup of the considered network consisting of a multiple antenna BS, a multiple antenna bidirectional RS and multiple users.}\vspace{-0.2cm}
     \label{system_fig}
 \end{center}
\end{figure}

 \begin{figure}[t]
 \begin{center}
     \includegraphics[width=\columnwidth]{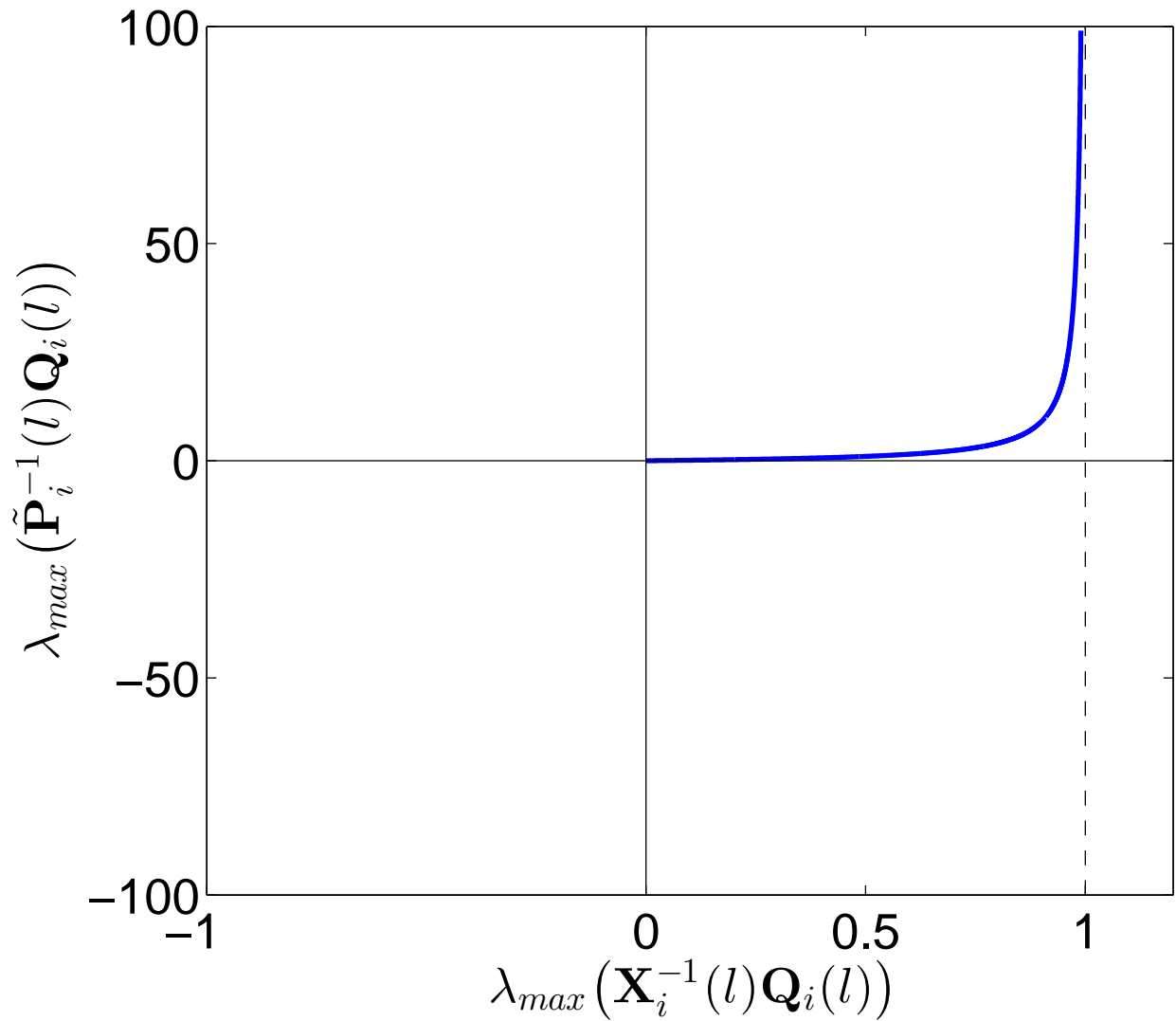}
     \caption{\small Maximum eigenvalue of matrix $\tilde{\mathbf{P}}_i^{-1}(l)\mathbf{Q}_i(l)$ against different values of $\lambda_\text{max} \big(\mathbf{X}_i^{-1}(l)\mathbf{Q}_i(l)\big)$. Note that both matrices are rank-1 and positive semidefinite.}\vspace{-0.2cm}
     \label{fig:lambda}
 \end{center}
\end{figure}

 \begin{figure}[t]
 \begin{center}
     \includegraphics[width=\columnwidth]{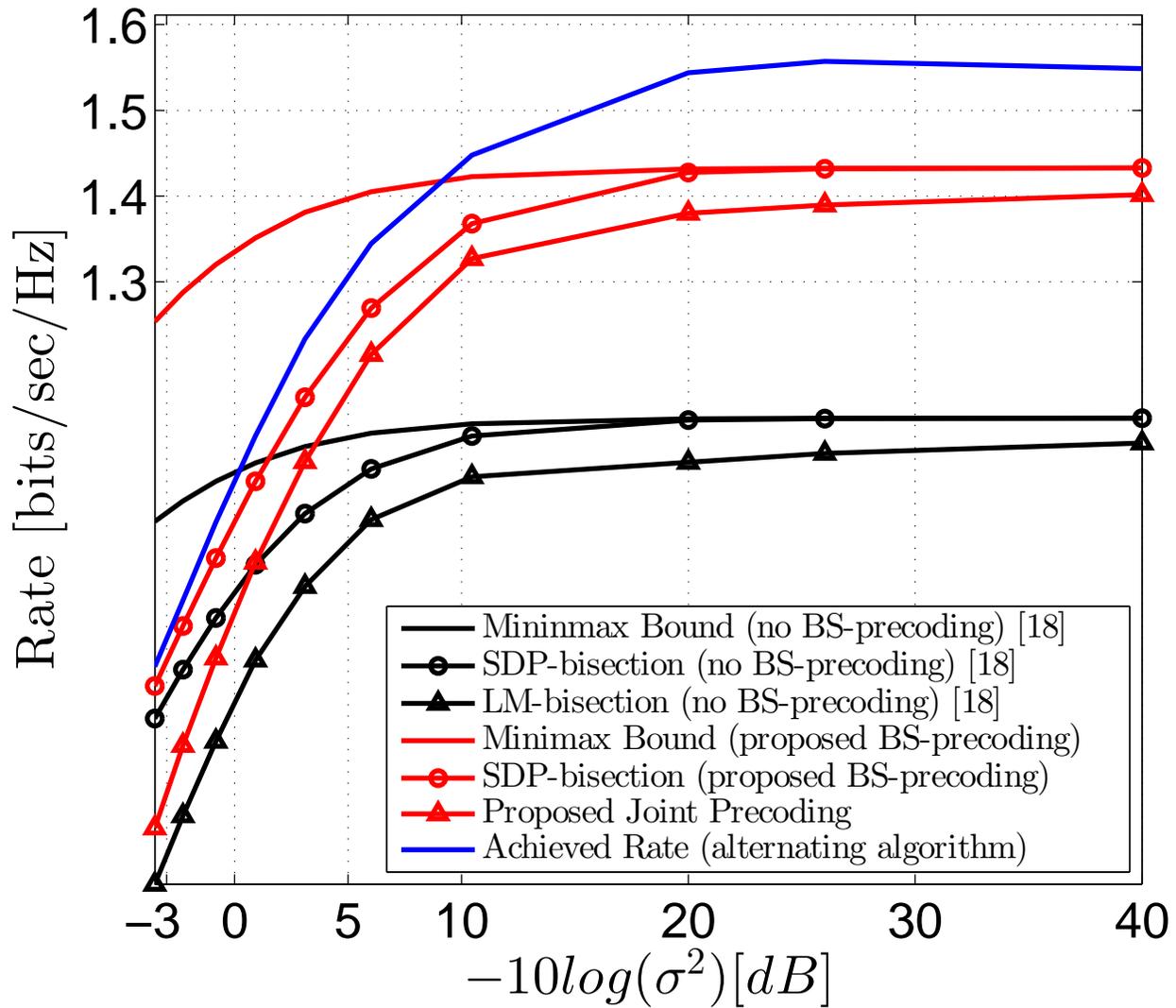}
     \caption{\small The minimum achievable rate amongst all downlink and uplink in different values of SNR.}
     \label{rate_fig}
 \end{center}
\end{figure}


 \begin{figure}[t]
 \begin{center}
     \includegraphics[width=\columnwidth]{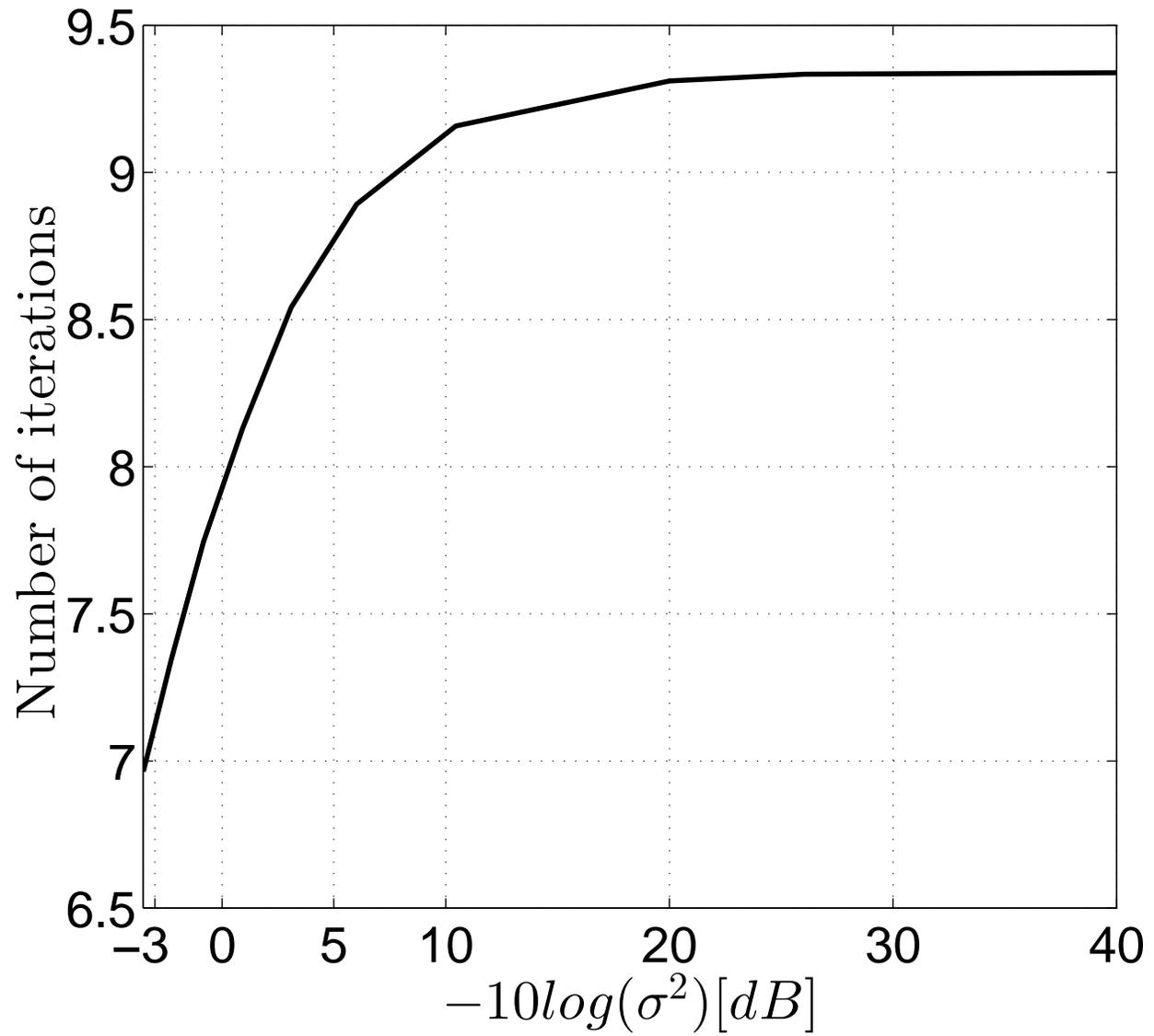}
     \caption{\small Number of iterations to find $\mathbf{W}$ using Algorithm \ref{Wt_algorithm} in different values of SNR.}
     \label{iter_Wt_fig}
 \end{center}
\end{figure}

 \begin{figure}[t]
 \begin{center}
     \includegraphics[width=\columnwidth]{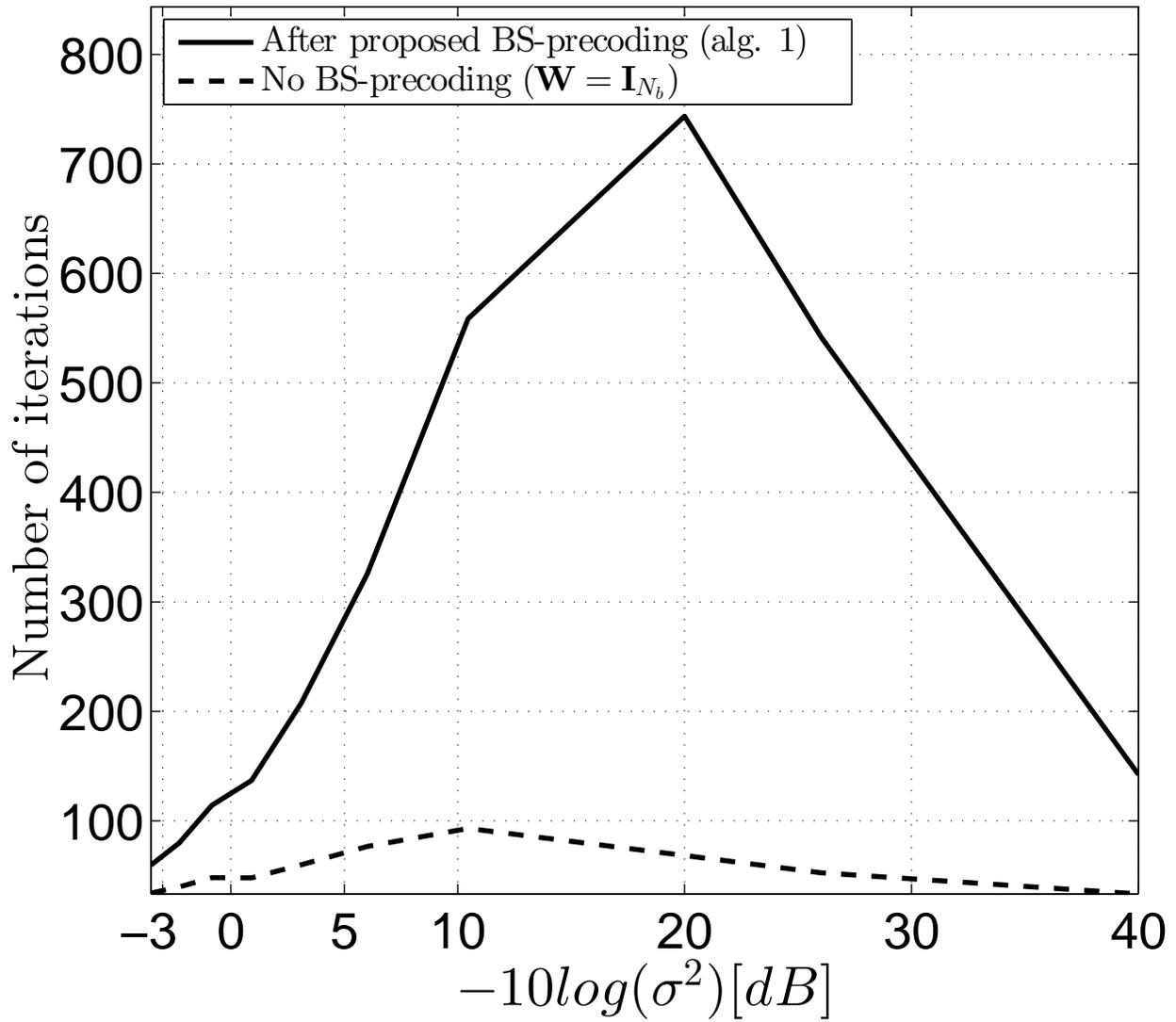}
     \caption{\small Total number of iterations in LM-bisection method versus SNR. Here the number of iterations is calculated by adding up number of iterations in each bisection step.}
     \label{LM_iter_fig}
 \end{center}
\end{figure}


\begin{algorithm} 
 \caption{Iterative Algorithm for finding sub-optimal solution of \eqref{MMF-BS1}}
\label{Wt_algorithm}
\begin{algorithmic}  
\STATE {\bf initialization:} $\boldsymbol{\Psi}(0)=\sqrt{\frac{P_B}{N_b}}\;\mathbf{I}_{N_b}$, $l \leftarrow 0$, $i \in \mathcal{S}'$ \\
\begin{align*}
&A_0=\mathcal{H}\big(\lambda_\text{max}(\tilde{\mathbf{P}}^{-1}_i \mathbf{Q}_i)\big),\;\mathbf{W}(0)=\sqrt{\frac{P_B}{N_b}}~~\mathbf{I}_{N_b} \\
&A_1=\mathcal{H}\big(\lambda_\text{max}(\tilde{\mathbf{P}}^{-1}_i \mathbf{Q}_i)\big),\;\mathbf{W}(0)=\sqrt{\frac{P_B}{\operatorname{Tr}(\boldsymbol{\Lambda}^{-1})}}~~ \mathbf{U}\mathbf{\Lambda}^{-1/2}.
\end{align*}
\IF {$A_0 \geq A_1$} 
\STATE $\mathbf{W}(0)=\sqrt{\frac{P_B}{N_b}}~~\mathbf{I}_{N_b}$ 
\ELSE
\STATE $\mathbf{W}(0)=\sqrt{\frac{P_B}{\operatorname{Tr}(\boldsymbol{\Lambda}^{-1})}}~~ \mathbf{U}\mathbf{\Lambda}^{-1/2}$
\ENDIF 
\WHILE{convergence}
\STATE $l \leftarrow l+1$ 
\STATE find $\mathbf{W}(l)$ with \eqref{Wt_closed_form} and update $\mathbf{Q}_i(l)$, $\tilde{\mathbf{P}}^{-1}_i(l)$ and $\mathbf{X}_i(l)$
\ENDWHILE 

\STATE {\bf return $\mathbf{W}$}.
\end{algorithmic}
\end{algorithm}

\begin{table}
\begin{center}
 \caption{Comparison of the complexity of the proposed algorithm with alternating method.} \label{complexity_table}
\begin{tabular}{ |c|c|c|c| }
\hline
Method &RS precoding & BS precoding & Total (for $N_r \geq N_b$)\\ \hline
\multirow{2}{*}{Alternating \cite{SIDIRO2006}}
 & SDP-bisection:  & SDP-bisection:  & $N_a.\,\mathcal{O}\big((N^{2}_r)^6\big) \times$\\
 & $\mathcal{O}\big((N^{2}_r)^6\big).\mathcal{O}\big(\sqrt{2N_b}N_r\log (1/\epsilon_{SDP})\big)$ & $\mathcal{O}\big((N^{2}_b)^6\big).\mathcal{O}\big(\sqrt{N_b}N_b \log (1/\epsilon_{SDP})\big)$ 
 &  $\mathcal{O}\big(\sqrt{2N_b}N_r\log (1/\epsilon_{SDP})\big)$ \\ \hline
\multirow{2}{*}{Proposed} 
 & LM-bisection \cite{DartmannTVT13}: & Algorithm \ref{Wt_algorithm}: &  \multirow{2}{*}{} \\
 & $N_{LM}.\mathcal{O}((N^{2}_r)^3)$ & $N_w.\mathcal{O}((N^2_r)^3)$  & $(N_w+N_{LM}).\mathcal{O}((N^{2}_r)^3)$ \\ \hline
\end{tabular}
\end{center}
\vspace{-1cm}
\end{table}
\end{document}